\documentclass[twoside]{aiml14}

\usepackage{aiml14macro}

\usepackage[english]{babel}
\usepackage{amsfonts}
\usepackage{ifthen}
\usepackage{xspace}
\usepackage{xcolor}
\usepackage{lmodern}
\usepackage{comment}
\usepackage{multicol}

\usepackage{amsmath}
\usepackage{amssymb}
\usepackage{mathtools}
\usepackage{stmaryrd}

\usepackage{multicol}
\usepackage[active]{srcltx}

\newcommand{\mathCommandFont}[1]{\mathrm{#1}}
\newcommand{\logicFont}[1]{\mathcal{#1}}
\newcommand{\classFont}[1]{\mathrm{#1}}

\newcommand{\halfliteral}[1]{\protect\ensuremath{#1}}
\newcommand{\literal}[1]{\halfliteral{#1}\xspace}
\newcommand{\bracketOperator}[3]{\literal{#1\nobreak#3\nobreak#2}}
\newcommand{\commandOperator}[2]{\literal{\mathord{\mathCommandFont{#1}\ifthenelse{\equal{#2}{}}{}{(\nobreak#2\nobreak)}}}}

\newcommand{\set}[3][]{\literal{\left\{#2\;\middle|\;\ifthenelse{\equal{#1}{}}{\text{#3}}{\parbox{#1}{#3}}\right\}}}

\newcommand{\powerset}[1]{\literal{\mathcal{P}(\nobreak#1\nobreak)}}

\newcommand{\ddfn}{\mathrel{\mathop{{\mathop:}{\mathop:}}}=}

\DeclareMathOperator{\N}{\ensuremath{\mathbb{N}}}

\newcommand{\dep}[1][\cdot]{\literal{\mathCommandFont{=}\ifthenelse{\equal{#1}{}}{}{(\nobreak#1\nobreak)}}}

\DeclareMathOperator*{\idis}{\ovee}

\DeclareMathOperator*{\midis}{\text{{\Large\ensuremath{\varovee}}}}
\DeclareMathOperator*{\Idis}{\text{{\huge\ensuremath{\varovee}}}}

\newcommand{\poc}[2][]{\bracketOperator{\left\lgroup}{\right\rgroup}{\ifthenelse{\equal{#1}{}}{#2}{\begin{array}{@{}#1@{}}#2\end{array}}}}

\newcommand{\logic}[1]{\literal{\logicFont{#1}}}
\newcommand{\paraLogic}[2]{\ensuremath{\logic{#1}\ifthenelse{\equal{#2}{}}{}{(#2)}}\xspace}

\newcommand{\ff}{\logic{F}}
  \xspaceaddexceptions{\ff}

\renewcommand{\iff}{\Leftrightarrow}

\newcommand{\kbisim}{\rightleftarrows_{\,k}}
\newcommand{\kBisim}{\;[\rightleftarrows_{\,k}]\;}

\newcommand{\Bisimk}{\;[\rightleftarrows_{\,k-1}]\;}
\newcommand{\bisimK}{\rightleftarrows_{\,k+1}}
\newcommand{\BisimK}{\;[\rightleftarrows_{\,k+1}]\;}

\newcommand{\mBisim}{\;[\rightleftarrows_{\,m}]\;}
\newcommand{\nbisim}{\rightleftarrows_{\,n}}
\newcommand{\nBisim}{\;[\rightleftarrows_{\,n}]\;}
\newcommand{\obisim}{\rightleftarrows_{\,0}}

\newcommand{\class}[1]{\literal{\classFont{#1}}}

\newcommand{\TIME}{\class{TIME}}
\newcommand{\NTIME}{\class{NTIME}}
\newcommand{\ATIME}{\class{ATIME}}
\newcommand{\SPACE}{\class{SPACE}}
\newcommand{\NSPACE}{\class{NSPACE}}
\newcommand{\ASPACE}{\class{ASPACE}}

\newcommand{\NEXP}{\class{NEXP}}
\newcommand{\PSPACE}{\class{PSPACE}}

\xspaceaddexceptions{\TIME \NTIME \ATIME \SPACE \NSPACE \ASPACE}

\newcommand{\calL}{\literal{\mathcal{L}}}

\newcommand{\EMDL}{\mathcal{EMDL}}
\newcommand{\MDL}{\mathcal{MDL}}

\newcommand{\ML}{\mathcal{ML}}

\newcommand{\cK}{\mathcal{K}}
\newcommand{\md}{\mathrm{md}}

\newcommand{\KT}{\mathcal{KT}}

\newcommand{\type}{\mathrm{tp}}
\newcommand{\tset}{\mathrm{Tp}}

\DeclareMathOperator{\Dim}{Dim}
\DeclareMathOperator{\occ}{occ}

\def\lastname{Hella, Luosto, Sano and Virtema}

\begin{document}

\begin{frontmatter}
  \title{The Expressive Power of Modal Dependence Logic}
  \author{Lauri Hella}\thanks{The research of Lauri Hella was partially funded 
by a Professor Pool grant awarded by the 
{Finnish Cultural Foundation}.
The research of Jonni Virtema was supported by grant 266260 of the Academy of Finland,
and grants by the Finnish Academy of Science and Letters and the University of Tampere.}\footnotemark[1]
  \author{Kerkko Luosto}\footnotemark[1]
  \author{Katsuhiko Sano}\thanks{The research of Katsuhiko Sano
  was partially supported by JSPS KAKENHI, Grant-in-Aid for Young Scientists 
  (B) 24700146.}\footnotemark[2]
  \author{Jonni Virtema}\footnotemark[1]
  \address{$^{1}$School of Information Sciences \\ University of Tampere}
  \address{$^2$School of Information Science \\ Japan Advanced Institute of Science and Technology}
  
  \begin{abstract}
We study the expressive power of various modal logics with team semantics. We show that exactly the properties of teams that are downward closed and closed under team $k$-bisimulation, for some finite $k$, are definable in modal logic extended with intuitionistic disjunction. Furthermore, we show that the expressive power of modal logic with intuitionistic disjunction and extended modal dependence logic coincide. Finally we establish that any translation from extended modal dependence logic into modal logic with intuitionistic disjunction  increases the size of some formulas exponentially. 
  \end{abstract}

 \end{frontmatter}

\section{Introduction}
Dependence is a central notion in many scientific disciplines. For example in physics there are dependences in experimental data. Decision theory is concerned with identifying the variables on which the result depends. Furthermore, dependences between attributes is a key notion in database theory. In order to express such dependences in a formal framework, Väänänen \cite{va07} introduced first-order dependence logic. Dependence logic is based on team semantics, in which the truth of formulas is evaluated in sets of assignments instead of single assignments. Team semantics was originally defined by Hodges \cite{Hodges97c} as a means to obtain compositional semantics for the independence-friendly logic of Hintikka and Sandu \cite{hisa89}. 


With the aim to import dependences and team semantics to modal logic Väänänen \cite{va08} introduced \emph{modal dependence logic} $\MDL$. In the context of modal logic a team is just a set of states in a Kripke model. Modal dependence logic extends standard modal logic with team semantics by modal dependence atoms, $\dep[p_1,\dots,p_n,q]$. The intuitive meaning of the formula $\dep[p_1,\dots,p_n,q]$ is that within a team the truth value of the proposition $q$ is functionally determined by the truth values of the propositions $p_1,\dots,p_{n}$. 

Modal dependence logic is a first step toward combining functional dependences and modal logic. The logic however lacks the ability to express temporal dependences, only propositional dependences can be expressed. This is due to the restriction that only proposition symbols are allowed in the dependence atoms of $\MDL$. To overcome this defect Ebbing et~al.~\cite{EHMMVV13} introduced the
\emph{extended modal dependence logic}, $\EMDL$, which is obtained from $\MDL$ by extending the scope of dependence atoms to arbitrary modal formulas, i.e., dependence atoms in $\EMDL$ are of the form $\dep[\varphi_1,\dots\varphi_n,\psi]$, where $\varphi_1,\dots,\varphi_n,\psi$ are $\ML$ formulas.

In recent years the research around modal dependence logic and other 
modal logics with team semantics has been active, see e.g.
\cite{EHMMVV13,EbbingL12,ebloya,galliani13,lohvo13,vollmer13,Sevenster:2009,fanthesis}.
An important logic, closely related to modal dependence logic, is modal logic with intuitionistic disjunction, $\ML(\idis)$.
It was already observed by Väänänen \cite{va08} that dependence atoms can be defined
by using the intuitionistic disjunction~$\idis$. Using this observation Ebbing et~al.~\cite{EHMMVV13} showed that in terms of expressiveness, $\EMDL$ is contained in $\ML(\idis)$.
However, it was left open, whether the containment is strict, or
whether $\EMDL$ and $\ML(\idis)$ are actually equivalent with respect to 
expressive power.

Team semantics is also meaningful in the context of purely propositional logics.
Propositional dependence logic was extensively studied in the recent 
Ph.D.~thesis of Fan Yang \cite{fanthesis}. As pointed out in \cite{fanthesis},
propositional dependence logic is closely related to the inquisitive logic of 
Groenendijk \cite{Groenendijk07}  (see also \cite{Ciardelli11,sano11}). 
Like in the team semantics of propositional dependence logic, in inquisitive logic 
the meaning of formulas is defined on sets of assignments for proposition symbols. 
Ciardelli \cite{ciardelli09} proved that inquisitive logic is
expressively complete in the sense that every downward closed
property of teams (over a finite set of proposition symbols) is definable by 
a formula of inquisitive logic. Thus, we can say that the set of connectives used in inquisitive
logic is complete in the same spirit as, e.g., $\{\lnot,\land\}$ is a complete set of
connectives for propositional logic. 
Fan Yang \cite{fanthesis} proved that the same expressive completeness result
holds for propositional dependence logic, and consequently, inquisitive logic and
propositional dependence logic
are equivalent with respect to expressive power. 

It is well known that
the expressive power of modal logic can be characterized via bisimulation: 
by the famous result of Gabbay and van Benthem,
a class $\cK$ of pointed Kripke models $(K,w)$ is definable by a formula of modal logic 
if and only if $\cK$ is closed under $k$-bisimulation, for some $k\in\mathbb{N}$. 
In this paper we prove a joint extension to this characterization 
and the characterization of the expressive power of
inquisitive logic and propositional dependence logic mentioned above.
We first define a canonical extension of bisimulation suitable for team semantics, 
called team bisimulation.  Then we show that a class $\cK$ of of pairs $(K,T)$,
where $K$ is a Kripke model and $T$ is a team,
is definable by a sentence of $\ML(\varovee)$ if and only if $\cK$ is downward 
closed and closed under team $k$-bisimulation, for some $k\in\mathbb{N}$. 

Furthermore, we show that the expressive
power of $\EMDL$ coincides with that of $\ML(\idis)$, thus answering the open problem
from \cite{EHMMVV13} mentioned above.  In particular, we obtain as a corollary that
the expressive power of $\EMDL$ is also characterized by downward closure and
closure under team $k$-bisimulation. Since team $k$-bisimulation is a natural
adaptation of $k$-bisimulation to the context of team semantics, this result shows
that $\EMDL$ can be regarded as a canonical extension of modal logic for 
expressing dependences between formulas.

In addition, we introduce two semantical invariants for formulas of $\ML(\idis)$ and
$\EMDL$, which we call lower dimension and upper dimension, respectively.
We show that the truth of a formula in a team of a Kripke model can be determined
by checking its truth on subteams of a fixed size $n$. The lower dimension of the 
formula in question is the least $n\in\N$ such that this holds. Thus, lower dimension
gives rise to a natural classification of formulas with respect to their semantical
complexity, and we believe that it can also be used for analyzing the computational 
complexity of the model checking problem of modal formulas.

The upper dimension of a formula 
is defined as the largest number of maximal teams satisfying the formula
in any fixed Kripke model. We prove that the lower dimension of any formula
is less than or equal to its upper dimension. Moreover, we show
that the upper dimension admits well-behaved compositionally defined estimates.
These estimates are very useful in establishing upper bounds for lower dimension
as well, since finding good estimates for the lower dimension directly seems
to be difficult. 

Finally, we use the upper dimension for proving that any translation from
$\EMDL$ into $\ML(\idis)$ increases the size of some formulas exponentially.
To prove this, we show that the upper dimension of a dependence atom 
$\dep[p_1,\dots,p_n,q]$ is $2^{2^n}$, while the upper dimension of
any $\ML(\idis)$-formula $\varphi$ is at most $2^d$, where $d$ is the number
of occurrences of $\idis$ in $\varphi$.

\section{Background}

In this section we first give the syntax and team semantics for the modal logics studied in the paper.
We then formulate the notions of definability and expressive power in team semantics.
Finally we recall the basic results concerning bisimulation and definability in the
context of standard Kripke semantics.

\subsection{Modal logics with team semantics}

The syntax of modal logic $\ML$ could be defined in any standard way. However, when we consider 
the extension of $\ML$ by dependence atoms, it is useful to assume that all formulas are
in \emph{negation normal form}, i.e., negations occur only in front of atomic propositions.
Thus, we define the syntax of $\ML$ as follows:
\begin{definition}\label{def:ML}
Let $\Phi$ be a set of proposition symbols. The set of formulas of $\ML(\Phi)$ is generated by the following grammar
\[
\varphi \ddfn p\mid \neg p \mid (\varphi \wedge \varphi) \mid (\varphi \vee \varphi) \mid \Diamond \varphi \mid \Box \varphi, 
\]
where $p\in\Phi$.
\end{definition}
In this article we consider three extensions of $\ML$: \emph{modal logic with intuitionistic disjunction}
$\ML(\idis)$, \emph{modal dependence logic} $\MDL$, and \emph{extended modal dependence logic}
$\EMDL$.

\begin{definition}
\begin{enumerate}
\item The syntax of modal logic with intuitionistic disjunction $\ML(\varovee)(\Phi)$ is obtained by extending the syntax of $\ML$ by the grammar rule
\[
\varphi \ddfn (\varphi\varovee\varphi).
\]
\item The syntax for modal dependence logic $\MDL(\Phi)$ is obtained by extending the syntax of $\ML$ by dependence atoms
\[
\varphi\ddfn\dep[p_1,\dots, p_n,q], 
\]
where $p_1,\dots,p_n,q\in \Phi$.
\item The syntax for extended modal dependence logic $\EMDL(\Phi)$ is obtained by extending the syntax of $\ML$ by dependence atoms
\[
\varphi\ddfn\dep[\psi_1,\dots, \psi_n,\theta], 
\]
where $\psi_1,\dots,\psi_n,\theta$ are $\ML$-formulas.
\end{enumerate}
\end{definition}

The notion of Kripke model is defined as usual. Thus,
if $\Phi$ is a set of proposition symbols, a \emph{Kripke model $K$ over $\Phi$} is a triple $K = (W, R, V)$, where $W$ is a set of \emph{states} or \emph{(possible) worlds}, $R\subseteq W\times W$ is an \emph{accessibility relation}, and $V$ is a \emph{valuation} $V\colon \Phi \to \powerset{W}$. 

The semantics of $\ML$ is usually defined on pointed Kripke models.
We write $K,w\models\varphi$ if $\varphi\in\ML(\Phi)$
is true in $w\in W$ according to the standard Kripke semantics. However, to give
a meaningful semantics for dependence atoms and intuitionistic disjunction, we
need to consider arbitrary sets of states instead of single states as points of evaluation.

\begin{definition}
Let $K=(W,R,V)$ be a Kripke model.
\begin{enumerate}
\item Any subset $T$ of $W$ is called a \emph{team} of $K$.
\item For any team $T\subseteq W$ we write
$R[T]=\{v\in W \mid \exists w\in T: wRv\}$ and 
$R^{-1}[T]=\{w\in W \mid \exists v\in T: wRv\}$.
\item For teams $T,S\subseteq W$ we write $T[R]S$ if $S\subseteq R[T]$
and $T\subseteq R^{-1}[S]$. 
\end{enumerate}
\end{definition}

Thus, $T[R]S$ holds if and only if for every $v\in S$ there is $w\in T$ such that $wRv$,
and for every $w\in T$ there is $v\in S$ such that $wRv$.
We are now ready to define \emph{team semantics} for the modal logics studied
in this paper.

\begin{definition}\label{semantics}
The semantics for $\ML$, $\ML(\varovee)$, $\MDL$, and $\EMDL$ is defined as follows.
\begin{align*}
K,T\models p  \quad\Leftrightarrow\quad& T\subseteq V(p).\\
%
K,T\models \neg p \quad\Leftrightarrow\quad& T\cap V(p)=\emptyset. \\
%
K,T\models \varphi\land\psi \quad\Leftrightarrow\quad& K,T\models\varphi \text{ and } K,T\models\psi.\\
K,T\models \varphi\lor\psi \quad\Leftrightarrow\quad& K,T_1\models\varphi \text{ and } 
K,T_2\models\psi\\
&\text{for some $T_1,T_2$ such that $T_1\cup T_2= T$}.\\
K,T\models \Diamond\varphi \quad\Leftrightarrow\quad& K,T'\models\varphi \text{ for some $T'$ such that $T[R]T'$}.\\
K,T\models \Box\varphi \quad\Leftrightarrow\quad& K,T'\models\varphi, \text{ where $T'=R[T]$.}\\
\intertext{For $\ML(\varovee)$ we have the following additional clause:}
K,T\models \varphi\varovee\psi \quad\Leftrightarrow\quad& K,T\models\varphi \text{ or } K,T\models\psi.\\
\intertext{For $\MDL$ and $\EMDL$ we have the following additional clause:}
K,T\models \dep[\psi_1,\dots,\psi_n,\theta] \quad\Leftrightarrow\quad& \forall w,v\in T: \bigwedge_{i=1}^{n}(K,\{w\}\models\psi_i \Leftrightarrow K,\{v\}\models\psi_i)\\
& \text{implies }(K,\{w\}\models\theta\Leftrightarrow K,\{v\}\models\theta).
\end{align*}
\end{definition}

Note in particular that $\dep[\theta]$ is a formula saying that the truth value of $\theta$
is constant in the given team: $K,T\models\dep[\theta]$ if and only if either $K,\{w\}\models\theta$
for all $w\in T$, or $K,\{w\}\not\models\theta$ for all $w\in T$.

The team semantics for basic modal logic $\ML$ can be reduced to the usual
Kripke semantics in the sense that a team $T$ satisfies a formula $\varphi$ if and only if every 
state in $T$ satisfies $\varphi$:  

\begin{proposition}[{\cite[Theorem 1]{Sevenster:2009}}]\label{ML extension}
Let $K$ be a Kripke model, $T$ a team of $K$, and 
$\varphi$ an $\ML(\Phi)$-formula. Then 
\[
	K,T\models \varphi \quad\iff\quad K,w\models\varphi \text{ for every } w\in T.
\]
In particular, $K,\{w\}\models \varphi \;\iff\; K,w\models\varphi$.
\end{proposition}

\subsection{Definability and expressive power}

A \emph{$\Phi$-model with a team} is a pair $(K,T)$, where $K$ is a Kripke model over
$\Phi$ and $T$ is a team of $K$. We denote by $\KT(\Phi)$ the class of 
$\Phi$-models with teams.
If $\calL$ is one of the logics $\ML,\ML(\idis),\MDL,\EMDL$, then
each formula $\varphi\in\calL(\Phi)$ \emph{defines} a class of $\Phi$-models with teams:
\[
	\Vert\varphi\Vert:=\{(K,T)\in\KT(\Phi)\mid K,T\models\varphi\}.
\]
A class $\cK\subseteq\KT(\Phi)$ is \emph{definable} in $\calL$, if there is a formula
$\varphi\in\calL(\Phi)$ such that $\cK=\Vert\varphi\Vert$.

If $\calL$ s a logic whose semantics is defined on Kripke models with teams, then
the \emph{expressive power} of $\calL$ is just the collection of classes $\Vert\varphi\Vert$,
$\varphi\in\calL$, that are definable in $\calL$. Accordingly, the expressive power 
of two such logics $\calL$ and $\calL'$ can be compared as follows:
\begin{itemize}
\item $\calL'$ is \emph{at least as expressive as} $\calL$, $\calL\le\calL'$, if for 
every $\varphi\in\calL(\Phi)$ there is $\psi\in\calL'(\Phi)$ such that 
$\Vert\varphi\Vert=\Vert\psi\Vert$.
\item $\calL$ is \emph{less expressive than} $\calL'$, $\calL<\calL'$, if
$\calL\le\calL'$, but $\calL'\not\le\calL$.
\item $\calL$ and $\calL'$ are \emph{equally expressive}, $\calL\equiv\calL'$, if
$\calL\le\calL'$ and $\calL'\le\calL$.
\end{itemize}

Clearly $\ML\le\MDL\le\EMDL$.
Väänänen \cite{va08} gave a translation from $\MDL$ to $\ML(\idis)$,
and extending this translation to $\EMDL$, it was proved in \cite{EHMMVV13}
that $\EMDL\le\ML(\idis)$. Furthermore, it is easy to see that
dependence atoms are not definable in $\ML$, and in \cite{EHMMVV13}
it was proved that the non-propositional dependence atom $\dep[\Diamond p]$
is not definable in $\MDL$. Summing up, the following relationships between
the logics $\ML$, $\MDL$, $\EMDL$ and $\ML(\idis)$ are known:

\begin{proposition}[\cite{EHMMVV13}]\label{comparison}
$\ML<\MDL<\EMDL\le\ML(\idis)$. 
\end{proposition} 

Moreover, it was proved in \cite{EHMMVV13} that $\EMDL\equiv\ML(\idis_\ML)$,
where $\ML(\idis_\ML)$ is the fragment of $\ML(\idis)$ that does not allow nesting of
the intuitionistic disjunction $\idis$. However, it was left as an open problem in
\cite{EHMMVV13} whether the expressive power of $\EMDL$ is strictly weaker 
than that of $\ML(\idis)$. 

For any formula $\varphi\in\calL(\Phi)$, the class $\Vert\varphi\Vert$ can be seen as its 
\emph{global meaning}. But it is also useful to consider the meaning of formulas 
\emph{locally}, i.e., with respect to a fixed Kripke model. For any Kripke model 
$K=(W,R,V)$ over $\Phi$, each formula $\varphi\in\calL(\Phi)$ \emph{defines} a set of
teams of $K$: 
\[
	\Vert\varphi\Vert^K:=\{T\subseteq W\mid K,T\models\varphi\}.
\]

Note that it follows from Proposition \ref{ML extension} that the set $\Vert\varphi\Vert^K$ is
\emph{downward closed} for all $\varphi\in\ML$:
\begin{equation*}
	(*)\qquad\text{if } T\in\Vert\varphi\Vert^K\text{ and }S\subseteq T\text{, then }S\in\Vert\varphi\Vert^K.
\end{equation*}
Although Proposition \ref{ML extension} fails for the extensions $\ML(\idis)$, $\MDL$ and $\EMDL$
of $\ML$, downward closure still holds for all of these logics. 
We say that a logic $\mathcal{L}$ is \emph{downward closed} if $(*)$ holds for every
formula $\varphi\in\mathcal{L}$.

\begin{proposition}[\cite{va08},\cite{ebloya}]\label{down-cl}
The logics $\MDL$, $\EMDL$ and $\ML(\idis)$ are downward closed.
\end{proposition}
\begin{proof}
For $\MDL$ and $\ML(\idis)$, downward closure was proved in \cite{va08} and 
\cite{ebloya}. For $\EMDL$, the claim follows from the fact that
$\EMDL\le\ML(\idis)$.
\end{proof}

\subsection{Bisimulation and definability in Kripke semantics}

It is well known that the expressive power of basic modal logic $\ML$ with respect to 
Kripke semantics can be completely characterized in terms of $k$-bisimulation.
Our aim is to give an analogous characterization for the expressive power
of $\ML(\idis)$ and $\EMDL$. For this purpose we need some basic concepts and
results related to $k$-bisimulation.

The \emph{modal depth} $\md(\varphi)$ of a formula of $\ML(\Phi)$
is defined in the obvious manner, i.e.,
$\md(p)=\md(\lnot p)=0$ for $p\in\Phi$, 
$\md(\varphi\land\psi)=\md(\varphi\lor\psi)=\max\{\md(\varphi),\md(\psi)\}$, and
$\md(\Diamond\varphi)=\md(\Box\varphi)=\md(\varphi)+1$.

A \emph{pointed $\Phi$-model} is a pair $(K,w)$ such that $K$ is a Kripke model over $\Phi$,
and $w$ is a state in $K$.
Let  $k$ be a natural number, and let $(K,w)$ and $(K',w')$ be pointed $\Phi$-models. 
We say that $(K,w)$ and $(K',w')$ are \emph{$k$-equivalent}, in symbols
$K,w\equiv_k K',w'$, if for every $\varphi\in\ML(\Phi)$ with $\md(\varphi)\le k$
\[
	K,w\models\varphi\quad\iff\quad K',w'\models\varphi.
\]

\begin{definition}
Let  $k\in\N$, and let $(K,w)$ and $(K',w')$ be pointed $\Phi$-models.
We write $K,w\kbisim K',w'$ if $(K,w)$ and $(K',w')$ are $k$-bisimilar.
The $k$-bisimilarity relation $\kbisim$ can be defined recursively as follows:
\begin{itemize}
\item $K,w\obisim K',w'$ if and only if the equivalence
\(
	K,w\models p\iff K',w'\models p
\) 
holds for all $p\in\Phi$.
\item $K,w\bisimK K',w'$ if and only if $K,w\obisim K',w'$, and
\begin{enumerate}
\item[-] for every $v\in R[w]$ there is $v'\in R'[w']$ such that $K,v\kbisim K',v'$, and
\item[-] for every $v'\in R'[w']$ there is $v\in R[w]$ such that $K,v\kbisim K',v'$.
\end{enumerate}
(Here $R[w]$ is a shorthand notation for $R[\{w\}]$. Thus, $v\in R[w]\;\iff\; wRv$.)
\end{itemize}
\end{definition}

A class $\cK$ of pointed $\Phi$-models is \emph{closed under $k$-bisimulation} if 
it satisfies the following condition:
\begin{itemize}
\item $(K,w)\in\cK$ and $K,w\kbisim K',w'$ implies that $(K',w')\in\cK$.
\end{itemize}

We will also make use of the fact that for every pointed $\Phi$-model $(K,w)$
and every $k\in\N$ there is a formula that characterizes $(K,w)$ completely
up to $k$-equivalence. These \emph{Hintikka formulas} (or \emph{characteristic formulas})
are defined as follows (see e.g. \cite{Goranko2007}):

\begin{definition}
Assume that $\Phi$ is a finite set of proposition symbols.
Let $k\in\N$ and let $(K,w)$ be a pointed $\Phi$-model. The \emph{$k$-th Hintikka formula}
$\chi^k_{K,w}$ of $(K,w)$ is defined recursively as follows:
\begin{itemize}
\item $\chi^0_{K,w}:=\bigwedge \{p\mid p\in \Phi, w\in V(p)\}\land\bigwedge\{\lnot p\mid p\in\Phi, w\not\in V(p)\}$.
\item $\chi^{k+1}_{K,w}:=\chi^k_{K,w}\land \bigwedge_{v\in R[w]}\Diamond\chi^k_{K,v}
\land\Box \bigvee_{v\in R[w]}\chi^k_{K,v}$.
\end{itemize}
\end{definition}

It is easy to see that $\md(\chi^k_{K,w})=k$, and $K,w\models\chi^k_{K,w}$
for every pointed $\Phi$-model $(K,w)$. Moreover, the Hintikka formula 
$\chi^k_{K,w}$ captures the essence of $k$-bisimulation:

\begin{proposition}\label{hintikka}
Let $\Phi$ be a finite set of proposition symbols, $k\in\N$, and $(K,w)$ and $(K',w')$ pointed $\Phi$-models.
Then
\[
	K,w\equiv_k K',w'\quad\iff\quad K,w\kbisim K',w'\quad\iff\quad K',w'\models\chi^k_{K,w}.
\]
\end{proposition}

The characterization for the expressive power of $\ML$ with respect to Kripke-semantics 
can now be stated as follows:

\begin{proposition}[{van Benthem}, {Gabbay}] 
Assume that $\Phi$ is a finite set of proposition symbols. A class $\cK$
of pointed $\Phi$-models 
is definable in $\ML$ if and only if there is $k\in\N$ such that $\cK$ is closed under $k$-bisimulation.
\end{proposition}

\section{$\ML(\idis)$ and team bisimulation}

In this section we prove a characterization for the expressive power of $\ML(\idis)$.
This characterization is based on a natural adaptation of the notion of $k$-bisimulation
to logics with team semantics.  

\subsection{Bisimulation in team semantics}

We start by defining $k$-bisimulation in the context of team semantics; the definition
is directly based on the $k$-bisimulation relation $\kbisim$ for Kripke semantics.

\begin{definition}
Let $(K,T),(K',T')\in\KT(\Phi)$ and $k\in\N$. 
We say that $K,T$ and $K',T'$ are \emph{team $k$-bisimilar} and write $K,T\kBisim K',T'$ if
\begin{enumerate}
\item for every $w\in T$ there exists some $w'\in T'$ such that $K,w\kbisim K,w'$, and
\item for every $w'\in T'$ there exists some $w\in T$ such that $K,w\kbisim K,w'$.
\end{enumerate}
\end{definition}

It is well known that $K,w\kbisim K',w'$ implies $K,w\nbisim K',w'$ for all $n\le k$.
Using this it is easy to prove that the same holds also for team $k$-bisimilarity:

\begin{lemma}
Let $(K,T),(K',T')\in\KT(\Phi)$ and $k\in\N$. If $K,T\kBisim K',T'$, then $K,T\nBisim K',T'$ for all $n\le k$.
\end{lemma}
 
We say that a class $\cK\subseteq\KT(\Phi)$ is \emph{closed under team $k$-bisimulation} if 
it satisfies the condition:
\begin{itemize}
\item $(K,T)\in\cK$ and 
$K,T\kBisim K',T'$ implies that $(K',T)\in\cK$.
\end{itemize}

The next lemma shows that team $k$-bisimulation satisfies the natural counterparts of 
the back-and-forth properties that we used in defining $\kbisim$, as well as a couple of other
useful properties related to team semantics.

\begin{lemma}\label{team-bis}
Let $k\in\N$, and assume that $(K,T),(K',T')\in\KT(\Phi)$ are such that $K,T\BisimK K',T'$. Then
\begin{enumerate}
\item  for every $S$ s.t. $T[R]S$ there is $S'$ s.t. $T'[R']S'$ and $K,S\kBisim K',S'$;
\item  for every $S'$ s.t. $T'[R']S'$ there is $S$ s.t. $T[R]S$ and $K,S\kBisim K',S'$;
\item  $K,S\kBisim K',S'$ for $S=R[T]$ and $S'=R'[T']$;
\item  for all $T_1,T_2\subseteq T$ s.t. $T=T_1\cup T_2$ there are
$T'_1,T'_2\subseteq T'$ s.t. $T'=T'_1\cup T'_2$, and $K,T_i\BisimK K',T'_i$ for $i\in\{1,2\}$.
\end{enumerate}
\end{lemma}

\begin{proof}
(i) Assume that $T[R]S$. We define 
$$
	S':=\{v'\in R'[T']\mid \exists v\in S: K,v\kbisim K',v' \}.
$$
We will first show that $K,S\kBisim K',S'$.
By the definition of $S'$, we have $\forall v'\in S'\exists v\in S: K,v\kbisim K',v'$.
On the other hand, if $v\in S$, then there is $w\in T$ such that $wRv$. Furthermore,
since $K,T\BisimK K',T'$,
there is $w'\in T'$ such that $K,w\bisimK K',w'$, whence by the definition of $\bisimK$,
there is $v'\in W'$ such that $w'R'v'$ and $K,v\kbisim K',v'$. By the definition of $S'$, $v'$ is in $S'$.
Thus we see that
$\forall v\in S\exists v'\in S': K,v\kbisim K',v'$. 

To see that $T'[R']S'$ holds, note first that $S'\subseteq R'[T']$ by its definition.
Assume then that $w'\in T'$. Since $K,T\BisimK K',T'$, there is
$w\in T$ such that $K,w\bisimK K',w'$. Furthermore, since $T[R]S$, there is $v\in S$
such that $wRv$, and consequently there is $v'\in R'[w']$ such that $K,v\kbisim K',v'$.
By the definition of $S'$ we have now $v'\in S'$. Thus we conclude that $w'\in R'^{-1}[S']$.

(ii) The claim is proved in the same way as (i).

(iii) If $v\in R[T]$, then there is $w\in T$ such that $wRv$. By the assumption $K,T\BisimK K',T'$,
there is $w'\in T'$ such that $K,w\bisimK K',w'$. Hence, there is $v'$ such that $w'R'v'$
and $K,v\kbisim K',v'$. As $w'R'v'$, we have $v'\in R'[T']$. Thus, we conclude that
$\forall v\in R[T]\exists v'\in R'[T']: K,v\kbisim K',v'$. Using a symmetrical argument, we 
see that $\forall v'\in R'[T']\exists v\in R[T]: K,v\kbisim K',v'$.

(iv) Let $T_1,T_2\subseteq T$ be such that $T=T_1\cup T_2$.
Define now 
$$
	T'_i:=\{w'\in T'\mid \exists w\in T_i: K,w\bisimK K',w'\},
$$ 
for $i\in\{1,2\}$. Then by the definition of $T'_i$, $\forall w'\in T'_i\exists w\in T_i: K,w\bisimK K',w'$.
On the other hand, if $w\in T_i$, then $w\in T$, whence there is $w'\in T'$ such that
$K,w\bisimK K',w'$. By the definition of $T'_i$, then $w'$ is in $T'_i$. Thus we conclude that
$\forall w\in T_i\exists w'\in T'_i: K,w\bisimK K',w'$, as desired.
\end{proof}

\subsection{Characterizing the expressive power of $\ML(\idis)$}

Our goal is to prove that definability in $\ML(\idis)$ can be characterized by downward closure
and closure under team $k$-bisimulation. We already know that all $\ML(\idis)$-definable classes
are downward closed (see Proposition \ref{down-cl}). The next step is to prove that
$\ML(\idis)$-definable classes are closed under team $k$-bisimulation for some $k$.

\begin{theorem}\label{bisim-cl}
Let $\Phi$ be a set of proposition symbols, and let $\cK\subseteq\KT(\Phi)$. 
If $\cK$ is definable in $\ML(\idis)$, then there
is a $k\in\N$ such that $\cK$ is closed under $k$-bisimulation.
\end{theorem}

\begin{proof}
Assume that $\varphi\in\ML(\idis)$. We prove by induction on 
$\varphi$ that the class $\Vert\varphi\Vert$ is closed under 
$k$-bisimulation, where $k=\md(\varphi)$. 
\begin{itemize}
\item Let $\varphi= p\in\Phi$, and assume that $K,T\models\varphi$ and $K,T\kBisim K',T'$
for $k=0$.
Then $K,w\models p$ for all $w\in T$, and
for each $w'\in T'$ there is $w\in T$ such that $K,w\obisim K',w'$.
Thus, for all $w'\in T'$, $K',w'\models p$, whence $K',T'\models\varphi$.
\item The case $\varphi= \lnot p$ is similar to the previous one.
\item Let $\varphi=\psi\lor\theta$, and assume that $K,T\models\varphi$ and $K,T\kBisim K',T'$,
where $k=\md(\varphi)=\max\{\md(\psi),\md(\theta)\}$. Then there are $T_1,T_2\subseteq T$
such that $T=T_1\cup T_2$, $K,T_1\models \psi$ and $K,T_2\models \theta$.

By Lemma \ref{team-bis}(iv), there are subteams $T'_1,T'_2\subseteq T'$ such that
$T'=T'_1\cup T'_2$ and $K,T_i\kBisim K',T'_i$ for $i\in\{1,2\}$,
whence $K,T_1\mBisim K',T'_1$ and $K,T_2 \nBisim K',T'_2$, where $m=\md(\psi)$
and $n=\md(\theta)$. By induction hypothesis, $K',T'_1\models\psi$ and
$K',T'_2\models\theta$. Thus, we conclude that $K',T'\models\varphi$.
\item The cases $\varphi=\psi\land\theta$ and $\varphi=\psi\idis\theta$ are straightforward.
\item Let $\varphi=\Diamond\psi$, and assume that $K,T\models\varphi$ and $K,T\kBisim K',T'$,
where $k=\md(\varphi)=\md(\psi)+1$. Then there is a team $S$ on $K$ such that 
$T[R]S$ and $K,S\models\psi$.
By Lemma \ref{team-bis}(i), there is a team $S'$ such that $T'[R']S'$ and
$K,S\Bisimk K',S'$. By induction hypothesis, $K',S'\models\psi$, and
consequently $K',T'\models\varphi$.
\item Let $\varphi=\Box\psi$, and assume that $K,T\models\varphi$ and $K,T\kBisim K',T'$,
where $k=\md(\varphi)=\md(\psi)+1$. Then $K,R[T]\models\psi$, and by Lemma \ref{team-bis}(iii),
$K,R[T]\Bisimk K',R'[T']$. Thus, by induction hypothesis, $K',R'[T']\models\psi$, and consequently
$K',T'\models\varphi$.
\end{itemize}
\end{proof}

Next we prove that downward closure and closure under team $k$-bisimulation are
together a sufficient condition for $\ML(\idis)$-definability.

\begin{theorem}\label{mlidis-def}
Let $\Phi$ be a finite set of proposition symbols and let $\cK\subseteq\KT(\Phi)$.
Assume that $\cK$ is downward closed and closed
under $k$-bisimulation for some $k\in\N$. Then $\cK$ is definable in $\ML(\idis)$.
\end{theorem}

\begin{proof}
Let $\varphi$ be the formula
$$
	\Idis_{(K,T)\in \cK}\;\bigvee_{w\in T}\;\chi^k_{K,w},
$$
where $\chi^k_{M,w}$ is the $k$-th Hintikka-formula of the pair $(K,w)$.
Note that since $\Phi$ is finite, there are only finitely many different Hintikka-formulas
$\chi^k_{K,w}$. Thus, the disjunction $\bigvee_{w\in T}$ and the intuitionistic disjunction
$\midis_{(K,T)\in \cK}$ in $\varphi$ are essentially finite, whence $\varphi\in\ML(\idis)$.
We will now prove that $\varphi$ defines $\cK$.

Assume first that $(K_0,T_0)\in\cK$. By Proposition \ref{ML extension}, 
$K_0,\{v\}\models\chi^k_{K_0,v}$ for each $v\in T_0$.
Thus, $K_0,T_0\models\bigvee_{w\in T_0}\;\chi^k_{K_0,w}$, and consequently, $K_0,T_0\models\varphi$.

Assume for the other direction that $K_0,T_0\models\varphi$. Then there is a pair $(K,T)\in\cK$ such that 
$K_0,T_0\models\bigvee_{w\in T}\;\chi^k_{K,w}$. Thus, there are subsets $T_w$, $w\in T$, of $T_0$
such that $T_0=\bigcup_{w\in T} T_w$, and $K_0,T_w\models\chi^k_{K,w}$. 
By Proposition~\ref{ML extension}, $K_0,v\models\chi^k_{K,w}$ for every $v\in T_w$.
Let $T'=\{w\in T\mid T_w\not=\emptyset\}$. Since $\cK$ is downward closed, we have
$(K,T')\in \cK$. Observe now that for every $v\in T_0$ there is $w\in T'$ such that 
$K_0,v\models\chi^k_{K,w}$,
and for every $w\in T'$ there is $v\in T_0$ such that $K_0,v\models\chi^k_{K,w}$.
By Proposition \ref{hintikka} this means that $K,T'\kBisim K_0,T_0$. Since  $\cK$ is closed under
$k$-bisimulation, we conclude that $(K_0,T_0)\in \cK$.
\end{proof}

Putting Proposition \ref{down-cl}, Theorem \ref{bisim-cl} and Theorem \ref{mlidis-def}
together, we finally get the promised characterization for the expressive power 
of $\ML(\idis)$.

\begin{corollary}\label{mlidis-exp}
A class $\cK\subseteq\KT(\Phi)$ is definable in $\ML(\idis)$ if and only if
$\cK$ is downward closed and  there exists $k\in\N$ such that $\cK$ is closed under $k$-bisimulation.
\end{corollary}

Note that from the proof of Theorem \ref{mlidis-def} we obtain the following
normal form for $\ML(\idis)$-formulas: every formula $\varphi\in\ML(\idis)$
is equivalent with a formula of the form $\midis \Psi$, where $\Psi$
is a finite set of $\ML$-formulas. This normal form was proved in \cite{lohvo13},
but the idea goes back to \cite{Sevenster:2009}. Note further that each formula
in $\Psi$ can be assumed to be a disjunction of Hintikka formulas 
$\chi^k_{K,w}$, where $k$ is the modal depth of $\varphi$.

\section{$\EMDL$ is equivalent to $\ML(\idis)$}\label{emdl-idis}

By Proposition \ref{comparison}, we know that $\ML(\idis)$ is at least as expressive
as $\EMDL$. In this section we show that the converse is also true, thus solving the
problem that was left open in \cite{EHMMVV13}. 

\begin{theorem}\label{idis-emdl}
$\ML(\idis)\le\EMDL$.
\end{theorem}

The proof we give for Theorem \ref{idis-emdl} is an adaptation of the proof in \cite{fanthesis} of the
corresponding result for propositional logic with intuitionistic disjunction and propositional 
dependence atoms. The main idea (Lemma \ref{non-subset}) is originally due to Taneli Huuskonen.

Before proving Theorem \ref{idis-emdl}, we introduce some auxiliary concepts, and
prove a couple of lemmas concerning them.

Let $\Psi$ be a finite set of $\ML(\Phi)$-formulas, and let $K$ be a Kripke model over 
$\Phi$ and $w$ a state in $K$.
The \emph{$\Psi$-type} of $w$ in $K$ is defined as 
$$
	\type_\Psi(K,w):=\{\psi\in\Psi\mid K,w\models\psi\}.
$$ 
Furthermore, the \emph{$\Psi$-type} of a team $T$ of $K$ is just the set of $\Psi$-types
of its elements:
$$
	\tset_\Psi(K,T):=\{\type_\Psi(K,w)\mid w\in T\}.
$$

Each $\Psi$-type $\Gamma\subseteq\Psi$ can be defined by a formula: 
Let 
$$
	\theta_\Gamma:=\bigwedge_{\psi\in\Gamma}\psi\land
	\bigwedge_{\psi\in\Psi\setminus\Gamma}\psi^\lnot
$$
where $\psi^\lnot$ denotes the formula obtained from $\lnot\psi$ by pushing the negations
in front of proposition symbols. Then it is easy to see that $\type_\Psi(K,w)=\Gamma$
if and only if $K,w\models\theta_\Gamma$.

\begin{lemma}\label{typesets}
Assume that $(K,T),(K',T')\in\KT(\Phi)$, and let $\Psi$ be a finite set of $\ML(\Phi)$-formulas.
\begin{enumerate}
\item For each $\psi\in\Psi$, 
$\;K,T\models\psi\;$ if and only if $\;\psi\in{\bigcap\tset_\Psi(K,T)}$.
\item If $K,T\models\midis\Psi$ and $\tset_\Psi(K',T')\subseteq\tset_\Psi(K,T)$,
then $K',T'\models\midis\Psi$.
\end{enumerate}
\end{lemma}

\begin{proof}
(i) If $K,T\models\psi$, then by Proposition \ref{ML extension}, $K,w\models\psi$ for every
$w\in T$, which means that $\psi\in\type_\Psi(K,w)$ for every $w\in T$.
On the other hand, if $\psi\in\bigcap\tset_\Psi(K,T)$, then $K,w\models\psi$ for every
$w\in T$. By Proposition \ref{ML extension}, it follows that $K,T\models\psi$.

(ii) Assume that $K,T\models\midis\Psi$ and $\tset_\Psi(K',T')\subseteq\tset_\Psi(K,T)$.
Thus, $K,T\models\psi$ for some $\psi\in\Psi$, and by claim (i), $\psi\in\bigcap\tset_\Psi(K,T)$.
Since $\tset_\Psi(K',T')\subseteq\tset_\Psi(K,T)$, it follows that $\psi\in\bigcap\tset_\Psi(K',T')$.
Thus, $K',T'\models\psi$, and consequently $K',T'\models\midis\Psi$.
\end{proof}

Consider next the formula $\gamma:=\bigwedge_{\psi\in\Psi}\dep[\psi]$. 
It says that the truth value of each $\psi$ in $\Psi$ is constant, whence
$K,T\models\gamma$ if and only if $|\tset_\Psi(K,T)|\le 1$.
Define now recursively 
$$
	\gamma^0:= p\land\lnot p, \qquad
	\gamma^{k+1}: =(\gamma^k\lor\gamma)
$$
It is straightforward to show by induction that for all $k\in\N$, $K,T\models\gamma^k$ if and only if 
$|\tset_\Psi(K,T)|\le k$.

\begin{lemma}\label{non-subset}
Let $\Psi$ be a finite set of $\ML(\Phi)$-formulas.
If $(K,T)\in\KT(\Phi)$, $T\not=\emptyset$, then there is a formula 
$\xi_{K,T}\in\EMDL(\Phi)$ such that for every
$(K',T')\in\KT(\Phi)$
\[
	K',T'\models\xi_{K,T}\quad\iff\quad\tset_\Psi(K,T)\not\subseteq\tset_\Psi(K',T').
\]
\end{lemma}

\begin{proof}
Let $|\tset_\Psi(K,T)|=k+1$. We define 
\[
	\xi_{K,T}:=\Bigl(\bigvee_{\Gamma\in X}
	\theta_\Gamma\Bigr)\lor\gamma^k,
\]
where $X=\powerset{\Psi}\setminus\tset_\Psi(K,T)$.
Now given a pair $(K',T')\in\KT(\Phi)$ we have
\begin{eqnarray*}
	K',T'\models\xi_{K,T}&\;\iff\;&\text{there are } T_1,T_2\text{ such that } T_1\cup T_2=T'\text{ and } \\
	&&\qquad\qquad\tset_\Psi(K',T_1)\subseteq X\text{ and }|\tset_\Psi(K',T_2)|\le k \\
	&\;\iff\;& |\tset_\Psi(K,T)\cap\tset_\Psi(K',T')|\le k\\
	&\;\iff\;& \tset_\Psi(K,T)\not\subseteq\tset_\Psi(K',T').
\end{eqnarray*}
\end{proof}

\begin{proof}{\bf of Theorem \ref{idis-emdl}}.
Let $\varphi$ be an $\ML(\idis)(\Phi)$-formula. By the normal form derived in the proof
of Theorem \ref{mlidis-def}, 
we may assume that $\varphi$ is of the form $\midis\Psi$,
where  $\Psi$ is a finite set of $\ML(\Phi)$-formulas.

Let $\eta$ be the formula
$$
	\bigwedge_{(K,T)\in\overline{\Vert\varphi\Vert}}\xi_{K,T},
$$ 
where $\overline{\Vert\varphi\Vert}=\KT(\Phi)\setminus\Vert\varphi\Vert$ and $\xi_{K,T}$ is as in 
Lemma \ref{non-subset}.
Since $\Psi$ is finite, there are finitely many different formulas of the form
$\xi_{K,T}$. Thus, the conjunction in $\eta$ is essentially finite, and hence $\eta$
is in $\EMDL$. 

To prove that $\Vert\eta\Vert=\Vert\varphi\Vert$, let $(K_0,T_0)\in\KT(\Phi)$. Assume first
that $(K_0,T_0)\in\Vert\varphi\Vert$, and consider any pair $(K,T)\in\overline{\Vert\varphi\Vert}$. 
It follows from Lemma~\ref{typesets} that
$\tset_\Psi(K,T)\not\subseteq\tset_\Psi(K_0,T_0)$, whence by Lemma \ref{non-subset},
$K_0,T_0\models\xi_{K,T}$. Thus we see that $(K_0,T_0)\in\Vert\eta\Vert$.

Assume then that $(K_0,T_0)\not\in\Vert\varphi\Vert$. Since 
$\tset_\Psi(K_0,T_0)
\subseteq\tset_\Psi(K_0,T_0)$, it follows from Lemma \ref{non-subset}
that $K_0,T_0\not\models\xi_{K_0,T_0}$. Thus we conclude that $(K_0,T_0)\not\in\Vert\eta\Vert$.
\end{proof}

Combining Proposition \ref{comparison} and Theorem \ref{idis-emdl}, we see
that the expressive power of $\EMDL$ and $\ML(\idis)$ coincide. This means 
that the characterization for the expressive power of $\ML(\idis)$ given in Corollary
\ref{mlidis-exp} is true for $\EMDL$, too.

\begin{corollary}
$\EMDL\equiv\ML(\idis)$.
\end{corollary}

\begin{corollary}\label{emdl-exp}
A class $\cK\subseteq\KT(\Phi)$ is definable in $\EMDL$ if and only if
$\cK$ is downward closed and  there is a $k\in\N$ such that $\cK$ is closed under $k$-bisimulation.
\end{corollary}

\section{Dimensions for modal formulas}

In this section we introduce two semantical invariants for formulas of
$\EMDL$ and $\ML(\idis)$. We will will first show that the truth of a formula $\varphi$
in a team $T$ of a Kripke model $K$ can be determined by considering only
subteams $T'\subseteq T$ of a fixed size $n$; we define the 
\emph{lower dimension} of $\varphi$ to be
the least $n$ such that this holds. Thus, lower dimension is a natural measure 
that can be used for classifying formulas with respect to their semantical complexity.
We also believe that lower dimension can be useful in analyzing the computational 
complexity of the model checking problem of modal formulas.

The other semantical invariant we introduce, the \emph{upper dimension} of a formula 
$\varphi$, is defined as the largest number of maximal teams $T$ that satisfy $\varphi$
in any single Kripke model $K$. We will show that the lower dimension of $\varphi$
is always less than or equal to the upper dimension. Moreover, we will show
that the upper dimension admits well-behaved estimates that are defined compositionally.
These estimates are very useful in establishing upper bounds for lower dimension
as well, since finding good estimates for the lower dimension directly is not
straightforward.

As we proved in the previous section, the expressive power of $\EMDL$ and $\ML(\idis)$ coincide.  However, there can be a considerable difference in the sizes of equivalent
formulas under any translation. 
It was already pointed out in \cite{EHMMVV13} that
there is an intrinsic difference in the complexity of $\EMDL$ and $\ML(\idis)$:
the satisfiability problem for the former is $\NEXP$-complete (\cite{EHMMVV13}), 
while for the latter it is $\PSPACE$-complete (\cite{Sevenster:2009}). 
This strongly hints to the possibility that there is no polynomially bounded
translation  from $\EMDL$ to $\ML(\idis)$.
Using the upper dimension, we will prove that this is indeed the case: 
any translation from $\EMDL$ to $\ML(\idis)$ introduces an exponential blow-up
for the size of formulas. 

\subsection{Lower and upper dimension}

Let $\varphi$ be a formula in $\ML(\idis)(\Phi)$, and let $n\in\N$. 
Adapting a notion that
was introduced by Jarmo Kontinen in \cite{jarmo} for first-order dependence logic,
we say that $\varphi$ is \emph{$n$-coherent} if the condition 
$$
	K,T\models\varphi\;\iff\; K,T'\models\varphi\text{ for all $T'\subseteq T$ 
	such that }|T'|\le n
$$
holds for all $(K,T)\in\KT(\Phi)$. 

It follows from Corollary \ref{mlidis-exp} 
that for every $\ML(\idis)(\Phi)$-formula $\varphi$ 
there is a natural number $n$ such that $\varphi$ is $n$-coherent. 
This can be seen as 
follows:  Let $k\in\N$ be such that $\Vert\varphi\Vert$ is closed under team $k$-bisimulation,
and let $n$ be the number of $\kbisim$-equivalence classes of pointed $\Phi$-models
$(K,w)$. If $K,T\models\varphi$, then by downward closure, $K,T'\models\varphi$ 
for every subteam $T'\subseteq T$. On the other hand, if $K,T\not\models\varphi$,
then $K,T'\not\models\varphi$ for any subteam $T'$ of $T$ such that for every
$w\in T$ there is $w'\in T'$ with $K,w\kbisim K,w'$. Clearly there is such a subteam
$T'$ with $|T'|\le n$.

Intuitively, the lower dimension of a formula $\varphi\in\ML(\idis)(\Phi)$ can be defined 
as the least $n$ such that $\varphi$ is $n$-coherent. However, due to technical reasons,
we formulate the definition of lower dimension in a bit different, but equivalent way.
Given a Kripke model $K$ over $\Phi$, let $N(\varphi,K)$ denote the family of minimal teams 
$T$ of $K$ such that $T\not\in\Vert\varphi\Vert^K$.

\begin{definition}
Let $\varphi\in\ML(\idis)(\Phi)$.  
The \emph{lower dimension} $\dim(\varphi)$ of $\varphi$ is the least $n\in\N$ such
that for every  Kripke model~$K$ over $\Phi$ and every $T\in N(\varphi,K)$
we have $|T|\le n$.
\end{definition}

We will next define the upper dimension for $\ML(\idis)$-formulas. 
Let $K$ be a Kripke model over $\Phi$ and let $\varphi$ an $\ML(\idis)(\Phi)$-formula.
As $\Vert \varphi \Vert^K$ is downward closed, it is natural to
study the family $M(\varphi,K)$ consisting of maximal elements of $\Vert \varphi \Vert^K$.
We will see below that $\Vert \varphi \Vert^K$ is \emph{generated} by  
$M(\varphi,K)$ in the sense that every team $T\in \Vert \varphi \Vert^K$ 
is contained in some team $S\in M(\varphi,K)$.  

\begin{definition}
 Let $\varphi\in\ML(\idis)(\Phi)$.  The 
\emph{upper dimension} $\Dim(\varphi)$ of $\varphi$ is the least $m\in\N$ such
that for every Kripke model~$K$ over $\Phi$ we have $|M(\varphi,K)|\le m$.
\end{definition}

Note that it is not a priori clear that the upper dimension is \emph{well-defined}:
if there is no uniform bound $m\in\N$ for the size of $M(\varphi,K)$ over
all Kripke models $K$, then $\Dim(\varphi)$ does not exist. In particular,
the definition of $\Dim(\varphi)$ requires that  $\Vert \varphi \Vert^K$ is
always \emph{finitely generated by} $M(\varphi,K)$, i.e., that $M(\varphi,K)$
is finite and generates $\Vert \varphi \Vert^K$ for all $K$. 

\begin{lemma} \label{dim-est}
$\Dim(\varphi)$ is well-defined for all $\varphi\in\ML(\idis)(\Phi)$. Moreover, 
we have the following estimates for $\varphi,\psi\in\ML(\idis)(\Phi)$:
\begin{multicols}{2}
\begin{enumerate}
\item 
$\Dim(p)=\Dim(\lnot p)=1$.

\item
$\Dim(\varphi\land\psi)\le \Dim(\varphi)\Dim(\psi)$.

\item
$\Dim(\varphi\lor\psi)\le \Dim(\varphi)\Dim(\psi)$.

\item
$\Dim(\varphi\idis\psi)\le \Dim(\varphi)+\nolinebreak\Dim(\psi)$.\nolinebreak

\item
$\Dim(\Diamond\varphi)\le\Dim(\varphi)$.

\item
$\Dim(\Box\varphi)\le\Dim(\varphi)$.

\end{enumerate}
\end{multicols}
\end{lemma}

\begin{proof}
We prove the first claim and the dimension estimates simultaneously
by induction on  $\varphi$.  Let $K=(W,R,V)$ be an arbitrary
Kripke model over $\Phi$. We omit the cases for (i), (iii) and (vi), since (i) is trivial, and (iii) and (vi) are analogous to (ii) and (v), respectively.

\begin{enumerate}
\setcounter{enumi}{1}
\item
We first notice that
   $
\Vert \varphi\land\psi \Vert^K = \Vert \varphi \Vert^K \cap \Vert \psi \Vert^K.
   $
By induction hypothesis, $\Vert \varphi \Vert^K$ and $\Vert \psi \Vert^K$
are finitely generated by $M(\varphi,K)$ and $M(\psi,K)$, respectively. 
Moreover, $|M(\varphi,K)|\le\Dim(\varphi)$ and
$|M(\psi,K)|\le\Dim(\psi)$.  It is immediate that
   $
M(\varphi\land\psi,K) \subseteq \{ T\cap U \mid T\in M(\varphi,K), U\in M(\psi,K) \}.
   $
   
Clearly, by the induction hypothesis the right-hand side of the inclusion above 
also generates the family $\Vert \varphi\land\psi \Vert^K$.
The inclusion now implies 
$|M(\varphi\land\psi,K)|\le |M(\varphi,K)\times M(\psi,K)|\le \Dim(\varphi)\Dim(\psi)$.
Hence, $\Dim(\varphi\land\psi)\le \Dim(\varphi)\Dim(\psi)$.

\setcounter{enumi}{3}
\item
For the intuitionistic disjunction, it holds that
   $$
M(\varphi\idis\psi,K) \subseteq M(\varphi,K) \cup M(\psi,K)
   $$
and the right-hand side of the inclusion generates the family
$\Vert \varphi\idis\psi \Vert^K$.  The dimension estimate follows immediately.

\item
For the diamond, we have that
   $
M(\Diamond\psi,K) \subseteq \{ R^{-1}[T] \mid T\in M(\varphi,K)\},
   $
and that $\{ R^{-1}[T] \mid T\in M(\varphi,K)\}$ generates  $\Vert \Diamond\psi\Vert^K$.
Thus we get that $|M(\Diamond\psi,K)|\le |M(\varphi,K)|$, which implies that $\Dim(\Diamond\varphi)\le\Dim(\varphi)$.
%
%
\end{enumerate}
\end{proof}

\begin{remark}
In \cite{ciardelli09}, Ciardelli gave estimates, that he calls \emph{Groenendijk's inequalities},
for the size of \emph{inquisitive meanings} of formulas. These estimates are essentially
equivalent to (i), (ii) and (iv) above.
In addition, he gave a similar estimate for the case of (intuitionistic) implication.
\end{remark}

The estimates given in Lemma~\ref{dim-est} are sharp in the sense that we cannot improve the upper bounds.
For conjunction (and implicitly also for the intuitionistic disjunction),
the following example demonstrates this sharpness.

\begin{example} \label{conjex}
Let $m$ and $n$ be positive integers.  We show that there are $\varphi,\psi\in\ML(\idis)$
such that 
$\Dim(\varphi)=m$, $\Dim(\psi)=n$ and $\Dim(\varphi\land\psi)=mn$.
Let $p_0,\ldots,p_{m-1},q_0,\ldots,q_{n-1}$ be distinct propositional symbols.
Put  
   $$
\varphi_i:=p_i\land \bigwedge_{k<m, k\neq i} \lnot p_k\quad\text{and}\quad
\psi_j:=q_j\land \bigwedge_{l< n, l\neq i}\lnot q_l,
   $$
for $i<m$ and  $j<n$.  Note that the formulas
$\varphi_i$, $i<m$, are satisfiable, but 
mutually contradictory in the classical sense, and similarly for $\psi_j$'s.
If $K=(W,R,V)$ is a Kripke model over $\{p_0,\ldots,p_{m-1},q_0,\ldots,q_{n-1}\}$,
then 
   $$
\Vert \varphi_i \Vert^K = \powerset{T_i} \quad\text{and}\quad
\Vert \varphi_j \Vert^K = \powerset{U_j}
   $$
for appropriate teams $T_i$ and $U_j$.
Clearly we can pick $K$ such that the intersections $T_i\cap U_j$ are all non-empty,
for $i<m$ and  $j<n$.
Define
   $$
\varphi:=\Idis_{i<m} \varphi_i \quad\text{and}\quad
\psi:=\Idis_{j<n} \psi_j.
   $$
The previous lemma gives the estimates $\Dim(\varphi)\le m$ and $\Dim(\psi)\le n$
for the upper dimensions.  However, in the Kripke model we have chosen,
   $$
\Vert \varphi \Vert^K = \bigcup_{i<m} \powerset{T_i}
\quad\text{and}\quad
\Vert \psi \Vert^K = \bigcup_{j<n} \powerset{U_j},
   $$
so $M(\varphi,K)=\{T_0,\ldots T_{m-1}\}$ and $M(\psi,K)=\{U_0,\ldots,U_{n-1}\}$,
which implies $\Dim(\varphi)=m$ and $\Dim(\psi)=n$.
Consider now the sentence $\varphi\land\psi$.  We have 
   $$
\Vert \varphi\land\psi \Vert^K 
= \bigcap_{i<m, j<n} \powerset{T_i\cap U_j},
   $$
so $M(\varphi\land\psi,K)=\{T_i\cap U_j \mid i<m, j<n\}$. 
Consequently, $\Dim(\varphi\land\psi)=mn$.
\end{example}

We will now prove that the upper dimension $\Dim(\varphi)$ is always a uniform
upper bound for $|N(\varphi,K)|$, whence $\dim(\varphi)$ is less than or equal to
$\Dim(\varphi)$.

\begin{lemma}
Assume that $\varphi\in\ML(\idis)(\Phi)$. Then $\dim(\varphi)\le \Dim(\varphi)$.
\end{lemma}

\begin{proof}
Let $K$ be a Kripke model, and let $U\in N(\varphi,K)$.
We need to prove that $|U|\le\Dim(\varphi)$ 
(if there are no such sets $U$, there is nothing to prove).  
For each $T\in M(\varphi,K)$, pick a state $w_T\in U\setminus T$.
Then the set $U_0=\{ w_T \mid T\in M(\varphi,K)\}$ is a subset of $U$, but not included
in any $T\in M(\varphi,K)$.  Hence, $U_0\in N(\varphi,K)$ and by the minimality of $U$,
we get $U=U_0$ and $|U|=|U_0|\le |M(\varphi,K)|\le \Dim(\varphi)$.
Hence, $\dim(\varphi)\le \Dim(\varphi)$.  
\end{proof}

The next example shows that the gap between upper and lower dimension may be arbitrarily large.

\begin{example}
For $j<n$, let the formulas $\psi_j$, as well as the Kripke model~$K$
and sets $U_j $, be as in Example~\ref{conjex},  Assume that $n\ge 4$.
To simplify notation, write $\psi_n=\psi_0$ and $U_n=U_0$. Consider the sentence
   $$
\theta:=\Idis_{j<n} (\psi_j \lor \psi_{j+1}).
   $$
Lemma~\ref{dim-est} gives the estimate $\Dim(\theta)\le n$.  In the Kripke model~$K$,
it is easy to see that $M(\theta,K)=\{U_j\cup U_{j+1} \mid j<n\}$.
Hence, $\Dim(\theta)=n$.  However, if a team~$T$ is such that $K,T\not\models\theta$,
then there is either a single point~$w\in T$ such that $K,\{w\}\not\models\theta$,
or there are $w\in U_j$, $w'\in U_k$ with $j\not\equiv k\pmod{n}$.
In the latter case, $K,\{w,w'\}\not\models\theta$.  The same reasoning applies to
other Kripke models than~$K$, so $\dim(\theta)=2$.
\end{example}

\subsection{The dimension of dependence atoms}

As $\EMDL\equiv\ML(\idis)$ and the definition of the upper and lower dimensions
is purely semantical, $\Dim(\varphi)$ and $\dim(\varphi)$ are defined for every 
$\EMDL$-formula~$\varphi$. Moreover, the estimates given in Lemma~\ref{dim-est}
are valid also for $\EMDL$-formulas.
For the modal dependence atoms, we have the following estimate 
for the upper dimension:
 
\begin{lemma}\label{depat-dim}
For the dependence atoms of $\EMDL(\Phi)$, we have that
   $
\Dim(\dep[\psi_1,\dots,\psi_n,\theta]) \le 2^{2^n}.
   $
Moreover, equality holds if $\psi_i$, $1\le i\le n$, and $\theta$ are distinct proposition symbols.
\end{lemma}

\begin{proof}
Denote the set $\{\psi_1,\ldots,\psi_n\}$ by $\Psi$ and the dependence atom  
$\dep[\psi_1,\dots,\psi_n,\theta]$ by $\varphi$. 
let $K=(W,R,V)$ be a Kripke model over $\Phi$, and let $X=\{\type_\Psi(K,w)\mid w\in W\}$,
where $\type_\Psi(K,w)$ is the $\Psi$-type of $w$ in $K$ (see Section \ref{emdl-idis}).
If $T\in M(\varphi,K)$, then there is a function
$f_T: X\to \{\bot,\top\}$ such that for all $w\in W$ 
$$
	M,w\models\theta\quad\iff\quad f_T(\type_\Psi(K,w))=\top.
$$

If $T$ and $U$ are different elements of $M(\varphi,K)$, then $T\cup U\not\in\Vert\varphi\Vert^K$,
whence there are states $w\in T$ and $u\in U$ such that $\type_\Psi(K,w)=\type_\Psi(K,u)$, but
$K,w\models\theta\;\iff\; K,u\not\models\theta$. This means that $f_T\not= f_U$. Thus, we see 
that $M(\varphi,K)$ has at most $ 2^{|X|}$ elements. Since $X\subseteq \powerset\Psi$
and $|\Psi|=n$, we arrive at the upper bound $2^{2^n}$ for $|M(\varphi,K)|$.

For the second claim, note that if $\psi_i\in\Phi$, $1\le i\le n$, and $\theta\in\Phi$ are distinct,
then there is a Kripke model such that every $\Gamma\subseteq\Psi$ is the $\Psi$-type
of some $w$ in $K$, and for every $f:X\to\{\bot,\top\}$ there is a team $T\in M(\varphi,K)$
such that $f=f_T$. Then $|X|=2^n$, and hence $|M(\varphi,K)|=2^{|X|}=2^{2^n}$.
\end{proof}

Thus, the upper dimension of dependence atoms can be doubly exponential with 
respect to the number of formulas occurring in it. On the other hand,  
any $\ML(\idis)$-formula can reach only single exponential upper dimension
with respect to its size. We prove this by considering 
the number $\occ_{\idis}(\varphi)$ of occurrences 
of $\idis$-symbols in the formula~$\varphi$.

\begin{proposition}\label{idis-occ}
Let $\varphi\in\ML(\idis)$. Then $\Dim(\varphi)\le 2^{\occ_{\idis}(\varphi)}$.  
\end{proposition}

\begin{proof}
The proof is a straightforward application of Lemma~\ref{dim-est} and induction.
For the literals, we have
   $$
\Dim(p)=\Dim(\lnot p)=1=2^0=2^{\occ_{\idis}(p)}=2^{\occ_{\idis}(\lnot p)}.
   $$
Suppose $\Dim(\varphi)\le 2^{\occ_{\idis}(\varphi)}$ and
$\Dim(\psi)\le 2^{\occ_{\idis}(\psi)}$.    
Then 
\begin{align*}
\Dim(\varphi\land\psi)& \le   \Dim(\varphi)\cdot\Dim(\psi) \\
                   &\le 2^{\occ_{\idis}(\varphi)}\cdot 2^{\occ_{\idis}(\psi)}
                      = 2^{\occ_{\idis}(\varphi)+\occ_{\idis}(\psi)} 
                      = 2^{\occ_{\idis}(\varphi\land\psi)},\\
\Dim(\varphi\lor\psi) &\le\Dim(\varphi)\cdot\Dim(\psi)\le  2^{\occ_{\idis}(\varphi)}\cdot 2^{\occ_{\idis}(\psi)}
   = 2^{\occ_{\idis}(\varphi\lor\psi)} \text{ and}\\
\Dim(\varphi\idis\psi) & \le \Dim(\varphi) + \Dim(\psi) \le  2^{\occ_{\idis}(\varphi)}+2^{\occ_{\idis}(\psi)} \\
                    & \le 2^{\occ_{\idis}(\varphi)}\cdot 2^{\occ_{\idis}(\psi)}+1 
                      \le 2^{\occ_{\idis}(\varphi)}\cdot 2^{\occ_{\idis}(\psi)}\cdot 2  \\ 
                    & =2^{\occ_{\idis}(\varphi)+\occ_{\idis}(\psi)+1} = 2^{\occ_{\idis}(\varphi\idis\psi)}.
\end{align*}
The case of the modal operators is trivial.
\end{proof}

\begin{theorem}
Assume that $\varphi\in\ML(\idis)$ is a formula such that 
$\Vert\varphi\Vert=\Vert\dep[p_1,\dots,p_{n},q]\Vert$. Then $\varphi$ 
contains more than $2^n$ symbols.
\end{theorem}

\begin{proof}
By Lemma~\ref{depat-dim},   $\Dim(\varphi)=\Dim(\dep[p_1,\dots,p_{n},q])=2^{2^n}$.
Thus, by Proposition \ref{idis-occ}, 
$2^{2^n}\le 2^{\occ_{\idis}(\varphi)}$ implying 
$2^n\le \occ_{\idis}(\varphi)$. This means that $\varphi$ contains at least $2^n$
intuitionistic disjunction symbols.
\end{proof}

Thus, any translation from $\EMDL$ to $\ML(\idis)$ necessarily leads to an exponential
blow-up in the size of formulas.

\section{Summary}

We studied the expressive power of various modal logics with team semantics: modal logic with intuitionistic disjunction $\ML(\varovee)$, modal dependence logic $\MDL$, and extended modal dependence logic $\EMDL$. We introduced the notion of team bisimulation and showed that a class $\cK$ of Kripke structures with teams is definable by a sentence of $\ML(\idis)$ if and only if $\cK$ is downward closed and closed under team $k$-bisimulation. In addition, we established that the expressive power of $\ML(\idis)$ and $\EMDL$ coincide and thus answered an open problem from \cite{EHMMVV13}. Furthermore, we introduced novel semantical invariants for formulas of $\EMDL$ and $\ML(\varovee)$, i.e., the notions of upper and lower dimension. By using these invariants, we obtained that the translations from $\MDL$ and $\EMDL$ into $\ML(\idis)$ are always worst-case exponential.

The characterization of the expressive power of $\EMDL$ and $\ML(\idis)$ gives rise to the question
whether similar characterizations can be found for other modal logics with team semantics. In particular,
is there such a characterization for the extension of $\ML$ with inclusion atoms or independence atoms?
For the definitions of these atoms, see the Ph.D. thesis \cite{fanthesis} of Fan Yang. 

\bibliographystyle{aiml14}
\bibliography{aiml14}

\begin{thebibliography}{10}
\expandafter\ifx\csname url\endcsname\relax
  \def\url#1{\texttt{#1}}\fi
\expandafter\ifx\csname urlprefix\endcsname\relax\def\urlprefix{URL }\fi
\newcommand{\enquote}[1]{``#1''}

\bibitem{ciardelli09}
Ciardelli, I., \enquote{Inquisitive Semantics and Intermediate Logics,}
  Master's thesis, University of Amsterdam (2009).

\bibitem{Ciardelli11}
Ciardelli, I. and F.~Roelofsen, \emph{Inquisitive logic}, J. Philosophical
  Logic \textbf{40} (2011), pp.~55--94.

\bibitem{EHMMVV13}
Ebbing, J., L.~Hella, A.~Meier, J.-S. M{\"u}ller, J.~Virtema and H.~Vollmer,
  \emph{Extended modal dependence logic}, in: \emph{WoLLIC}, 2013, pp.
  126--137.

\bibitem{EbbingL12}
Ebbing, J. and P.~Lohmann, \emph{Complexity of model checking for modal
  dependence logic}, in: M.~Bieliková, G.~Friedrich, G.~Gottlob,
  S.~Katzenbeisser and G.~Turán, editors, \emph{SOFSEM},  Lecture Notes in
  Computer Science  \textbf{7147} (2012), pp. 226--237.

\bibitem{ebloya}
Ebbing, J., P.~Lohmann and F.~Yang, \emph{Model checking for modal
  intuitionistic dependence logic}, in: G.~Bezhanishvili, S.~L\"obner, V.~Marra
  and F.~Richter, editors, \emph{Logic, Language, and Computation},  Lecture
  Notes in Computer Science  \textbf{7758}, Springer, 2013 pp. 231--256.

\bibitem{galliani13}
Galliani, P., \emph{The dynamification of modal dependence logic}, Journal of
  Logic, Language and Information \textbf{22} (2013), pp.~269--295.

\bibitem{Goranko2007}
Goranko, V. and M.~Otto, \emph{Model theory of modal logic}, in: P.~Blackburn,
  J.~{Van Benthem} and F.~Wolter, editors, \emph{Handbook of Modal Logic},
  Studies in Logic and Practical Reasoning  \textbf{3}, Elsevier, 2007 pp.
  249--329.

\bibitem{Groenendijk07}
Groenendijk, J., \emph{Inquisitive semantics: Two possibilities for
  disjunction}, in: P.~Bosch, D.~Gabelaia and J.~Lang, editors, \emph{TbiLLC},
  Lecture Notes in Computer Science  \textbf{5422} (2007), pp. 80--94.

\bibitem{hisa89}
Hintikka, J. and G.~Sandu, \emph{Informational independence as a semantical
  phenomenon}, in: \emph{Logic, methodology and philosophy of science, {VIII}
  ({M}oscow, 1987)},  Stud. Logic Found. Math.  \textbf{126}, North-Holland,
  Amsterdam, 1989 pp. 571--589.

\bibitem{Hodges97c}
Hodges, W., \emph{Compositional semantics for a language of imperfect
  information}, Logic Journal of the IGPL \textbf{5} (1997), pp.~539--563.

\bibitem{jarmo}
Kontinen, J., \enquote{Coherence and Complexity in Fragments of Dependence
  Logic,} Ph.D. thesis, University of Amsterdam (2010).

\bibitem{lohvo13}
Lohmann, P. and H.~Vollmer, \emph{Complexity results for modal dependence
  logic}, Studia Logica \textbf{101} (2013), pp.~343--366.

\bibitem{vollmer13}
Müller, J.-S. and H.~Vollmer, \emph{Model checking for modal dependence logic:
  An approach through post's lattice}, in: L.~Libkin, U.~Kohlenbach and
  R.~Queiroz, editors, \emph{Logic, Language, Information, and Computation},
  Lecture Notes in Computer Science  \textbf{8071}, 2013 pp. 238--250.

\bibitem{sano11}
Sano, K., \emph{First-order inquisitive pair logic}, in: M.~Banerjee and
  A.~Seth, editors, \emph{Logic and Its Applications},  Lecture Notes in
  Computer Science  \textbf{6521}, 2011 pp. 147--161.

\bibitem{Sevenster:2009}
Sevenster, M., \emph{Model-theoretic and computational properties of modal
  dependence logic}, J. Log. Comput. \textbf{19} (2009), pp.~1157--1173.

\bibitem{va07}
V{\"a}{\"a}n{\"a}nen, J., \enquote{Dependence Logic - A New Approach to
  Independence Friendly Logic,}  London Mathematical Society student texts
  \textbf{70}, Cambridge University Press, 2007.

\bibitem{va08}
V{\"a}{\"a}n{\"a}nen, J., \emph{Modal dependence logic}, in: K.~R. Apt and
  R.~van Rooij, editors, \emph{New Perspectives on Games and Interaction},
  Texts in Logic and Games  \textbf{4}, 2008 pp. 237--254.

\bibitem{fanthesis}
Yang, F., \enquote{On Extensions and Variants of Dependence Logic,} Ph.D.
  thesis, University of Helsinki (2014).

\end{thebibliography}
\clearpage

\end{document}